\documentclass[sigconf]{acmart}
\AtBeginDocument{%
  }

\copyrightyear{2026}
\acmYear{2026}
\setcopyright{cc}
\setcctype{by}
\acmConference[WWW '26]{Proceedings of the ACM Web Conference 2026}{April 13--17, 2026}{Dubai, United Arab Emirates}
\acmBooktitle{Proceedings of the ACM Web Conference 2026 (WWW '26), April 13--17, 2026, Dubai, United Arab Emirates}
\acmPrice{}
\acmDOI{10.1145/3774904.3792843}
\acmISBN{979-8-4007-2307-0/2026/04}

\usepackage{threeparttable}
\usepackage{multirow}
\usepackage{subcaption}
\usepackage{booktabs}
\usepackage{enumitem} 
\usepackage{graphicx}

\acmSubmissionID{97}

\begin{document}

\title{Hierarchical Semantic RL: Tackling the Problem of Dynamic Action Space for RL-based Recommendations}

\author{Minmao Wang}
\authornote{This work was done during an internship at Kuaishou Technology.}
\affiliation{%
  \institution{Fudan University}
  \city{Shanghai}
  \country{China}
}
\email{mmwang25@m.fudan.edu.cn}

\author{Xingchen Liu}
\affiliation{%
  \institution{Kuaishou Technology}
  \state{Beijing}
  \country{China}
}
\email{liuxingchen07@kuaishou.com}

\author{Shijie Yi}
\affiliation{%
  \institution{Kuaishou Technology}
  \city{Beijing}
  \country{China}
}
\email{yishijie@kuaishou.com}

\author{Likang Wu}
\authornote{Corresponding authors.}
\affiliation{%
  \institution{Tianjin University}
  \city{Tianjin}
  \country{China}
}
\email{wulk@tju.edu.cn}

\author{Hongke Zhao}
\affiliation{%
  \institution{Tianjin University}
  \city{Tianjin}
  \country{China}
}
\email{hongke@tju.edu.cn}

\author{Fei Pan}
\affiliation{%
  \institution{Kuaishou Technology}
  \city{Beijing}
  \country{China}
}
\email{panfei05@kuaishou.com}

\author{Qingpeng Cai}
\authornotemark[2] 
\affiliation{%
  \institution{Kuaishou Technology}
  \city{Beijing}
  \country{China}
}
\email{caiqingpeng@kuaishou.com}

\author{Peng Jiang}
\affiliation{%
  \institution{Kuaishou Technology}
  \city{Beijing}
  \country{China}
}
\email{jiangpeng@kuaishou.com}

\renewcommand{\shortauthors}{Minmao Wang et al.}

\begin{abstract}

Recommender Systems (RS) are fundamental to modern online services. While most existing approaches optimize for short-term engagement, recent work has begun to explore reinforcement learning (RL) to model long-term user value. However, these efforts face significant challenges due to the vast, dynamic action spaces inherent in RS, which hinder stable policy learning. To resolve this bottleneck, we introduce Hierarchical Semantic RL (HSRL), which reframes RL-based recommendation over a fixed Semantic Action Space (SAS). HSRL encodes items as Semantic IDs (SIDs) for policy learning, and maps SIDs back to their original items via a fixed lookup during execution. To align decision-making with SID generation, the Hierarchical Policy Network (HPN) operates in a coarse-to-fine manner, employing hierarchical residual state modeling to refine each level's context from the previous level’s residual, thereby reducing representation–decision mismatch. In parallel, a Multi-level Critic (MLC) provides token-level value estimates, enabling fine-grained credit assignment. Across public benchmarks and a large-scale production dataset from a leading short-video advertising platform, HSRL consistently surpasses state-of-the-art baselines. In online deployment over a 7-day A/B testing, it delivers an 18.421\% ADVV lift and a 1.251\% increase in Revenue, supporting HSRL as a scalable paradigm for RL-based recommendation.

\end{abstract}


\maketitle

\section{Introduction}

Recommender Systems (RS)~\cite{cr1,cr2,cr3,cr4,cr5,cr6,cr7} have become indispensable infrastructure in the modern digital economy, fundamentally driving growth across e-commerce, content distribution, and online advertising platforms. Their core mission is to achieve highly personalized content delivery by deeply profiling user behavior, ultimately maximizing long-term cumulative user utility. However, prevailing recommendation methodologies predominantly rely on the Supervised Learning (SL) paradigm, modeling recommendation as a static prediction task. The objectives are typically confined to optimizing immediate engagement metrics, such as click-through rate (CTR) or conversion rate~\cite{deepfm,din,mmoe,esmm}. This inherent myopic single-step optimization strategy ignores the sequential nature and long-term dynamic dependencies in user-system interactions. Consequently, these models often converge to a short-term local optimum, potentially sacrificing overall long-term systemic benefits. In recent years, Reinforcement Learning (RL)~\cite{ppo,a2c}, with its capability for sequential decision-making and maximizing cumulative reward, has emerged as a promising paradigm for addressing the long-term optimization challenge in RS. By explicitly modeling user–system interactions as a sequential process, RL enables long-horizon optimization beyond short-term engagement, offering a principled solution where SL falls short.

Nevertheless, as shown in Figure~\ref{fig:firstgraph}, applying deep RL policies to industrial-scale recommendation environments encounters a critical technical bottleneck: a highly dynamic and extremely large action space. Online RS are required to select optimal actions from a candidate item pool that scales into the millions or even billions and is constantly expanding. Within the RL framework, each item corresponds to a unique discrete action, making direct policy optimization infeasible due to the enormous cardinality and dynamic nature of the action set. This problem is the primary barrier preventing the deployment of mainstream RL algorithms in large-scale RS. Prior studies, such as continuous "hyper-actions" \cite{hac,ddpgra}, offer a workaround; however, the mapping from continuous policy outputs to discrete items often lacks differentiability and consistency, leading to suboptimal exploration and unstable training.

To effectively address the challenges of dynamic, high-dimensional action spaces, we leverage Semantic IDs (SIDs) \cite{tiger} to represent items in a tokenized, hierarchical semantic space, enabling structured action encoding. Building upon this, we propose the HSRL framework, the first Reinforcement Learning approach to utilize SIDs for establishing the novel Semantic Action Space (SAS) paradigm. Specifically, SAS deterministically projects the dynamic, high-dimensional item-level action space onto a fixed-size, low-dimensional hierarchical discrete space. Under this framework, the native RL policy network is exclusively tasked with exploration and optimization within this fixed-dimension semantic space. This design achieves an effective decoupling of the policy decision space from the item execution space, ensuring that the policy output dimension is no longer constrained by the item catalog size. This effectively resolves the action space dynamism problem, paving a new, scalable pathway for industrial RL recommendation systems.

\begin{figure}
    \centering
    \includegraphics[width=1\linewidth]{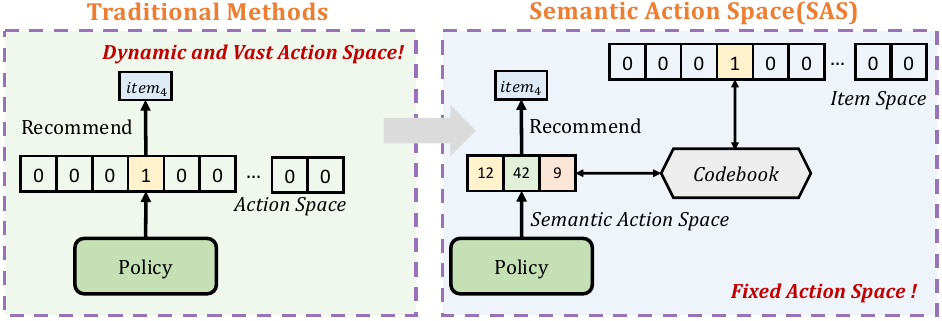}
    \caption{Action Space in Recommendation.}
    \label{fig:firstgraph}
\end{figure}

While the Semantic Action Space offers substantial scalability advantages, adopting structured SIDs as the RL action introduces two new major challenges. First, the SID generation process is a multi-step sequential autoregressive process with an inherent hierarchical dependency: the selection of higher-level tokens dictates the search scope for lower-level tokens, necessitating that the policy precisely captures and models this "coarse-to-fine" dependency, so that decisions for tokens at different levels cannot be treated equally. Second, user interactions in RS typically yield only a sparse, sequence-level reward signal after the entire SID sequence has been generated and executed. This signal setting significantly exacerbates the credit assignment problem. The diluted reward signal across the long decision chain results in unstable policy gradient estimation, severely hindering optimization efficiency and stability.

To effectively address these challenges, we introduce an innovative hierarchical Actor-Critic architecture featuring two core modules. First, the Hierarchical Policy Network (HPN) strictly aligns with the SID's hierarchical structure, generating tokens autoregressively. We introduce Hierarchical Residual State Modeling, ensuring that each layer's policy input recursively fuses the global context with the residual state information remaining after the preceding token's determination, thus guiding the subsequent fine-grained selection. This progressively refined decision mechanism enhances the policy's structured generalization within the semantic space, improving generalization across semantic paths. Second, the Multi-level Critic (MLC) provides independent, adaptive value estimation for every token decision within the SID generation sequence. This approach enables value decomposition and hierarchical credit assignment based solely on the single sequence-level reward. By efficiently breaking down the global reward into fine-grained local learning signals, the MLC alleviates instability in policy learning and ensures stable convergence across the long decision chains.

Our main contributions are summarized as follows:

\begin{enumerate}

\item \textbf{Paradigm Innovation:} We propose the HSRL framework, the first Reinforcement Learning approach to utilize SIDs for establishing the novel Semantic Action Space paradigm, resolving the action space explosion and dynamism challenge in RL-based recommender systems.

\item \textbf{Architectural Design:} We propose a novel architecture, including the Hierarchical Policy Network and Multi-level Critic, which jointly address the challenges of modeling token dependencies and credit assignment.
\item \textbf{Comprehensive Validation:} HSRL outperforms state-of-the-art baselines across public benchmarks and a large-scale production dataset from a leading Chinese short-video advertising platform. In a seven-day A/B testing, it achieves an 18.421\% ADVV lift and a 1.251\% increase in Revenue, demonstrating its strong effectiveness.
\end{enumerate}

\section{Related Work}

\textbf{Session-based recommendation.}
Session-based recommendation (SBR) is a specialized paradigm within the broader context of sequential recommendation (SR) \cite{Wang2021SRSurvey}. SR predicts the next item from a user’s interaction history, whereas SBR focuses on sequences with explicit beginning and termination boundaries, namely self-contained user sessions \cite{Liu2018STAMP}. Departing from myopic next-step accuracy, our goal is to maximize the future reward over the remainder of the session. Methodologically, SBR inherits and adapts the SR toolkit from collaborative filtering \cite{Koren2009MF,Rendle2010FM} to early deep sequence models such as GRU4Rec \cite{Hidasi2016GRU4Rec} and NARM \cite{Li2017NARM}, and further to Transformer-based architectures including SASRec \cite{Kang2018SASRec} and BERT4Rec \cite{Sun2019BERT4Rec}. A central challenge in SBR is constructing robust user-state representations under data sparsity and shifting intentions \cite{Chen2019GAUM,zhao2017deep}. To enrich structural signals, graph-based methods such as SR-GNN incorporate item–item relations \cite{Wu2019SRGNN}, and recent foundation-model approaches such as P5~\cite{Gururangan2022P5} and UniRec~\cite{Hou2023UniRec} cast recommendation as text generation. Although these methods may encode long-range interaction histories effectively, they do not directly optimize long-term user rewards.

\noindent\textbf{Reinforcement learning in Recommendation.} The application of Reinforcement Learning (RL) fundamentally transforms Recommender Systems (RS) by framing the sequential recommendation procedure as a Markov Decision Process (MDP) \cite{rl1,rl2,rl3,rl4,rl5,rl6,rl7,rl8,rl9,pei2019value, taghipour2007usage}. This formulation enables RL-based RS to learn an optimal recommendation policy through continuous, dynamic interaction with the user, thereby maximizing the long-term cumulative value or user experience \cite{afsar2021survey, shani2005mdp, sutton2018rl}. This paradigm offers a crucial advantage: the capacity to dynamically adapt to evolving user preferences, moving beyond static optimization metrics. Despite substantial theoretical and empirical advancements, practical RL-based RS still encounters several critical challenges. Foremost among these is the combinatorial explosion of state and action spaces \cite{dulac2015deep, ie2019slateq, liu2020state}, which severely impedes the convergence and scalability of conventional RL algorithms. Prior work, such as HAC \cite{hac}, attempts to mitigate this complexity by mapping the high-dimensional action space into a more tractable, lower-dimensional latent space. However, this dimensional reduction inherently introduces an information loss gap during the mapping process, a deficiency that may compromise recommendation performance. This challenge motivates our work to develop a more robust and information-preserving action representation learning mechanism.

\section{Methodology}

\begin{figure*}[t]
    \centering
    \includegraphics[width=\linewidth]{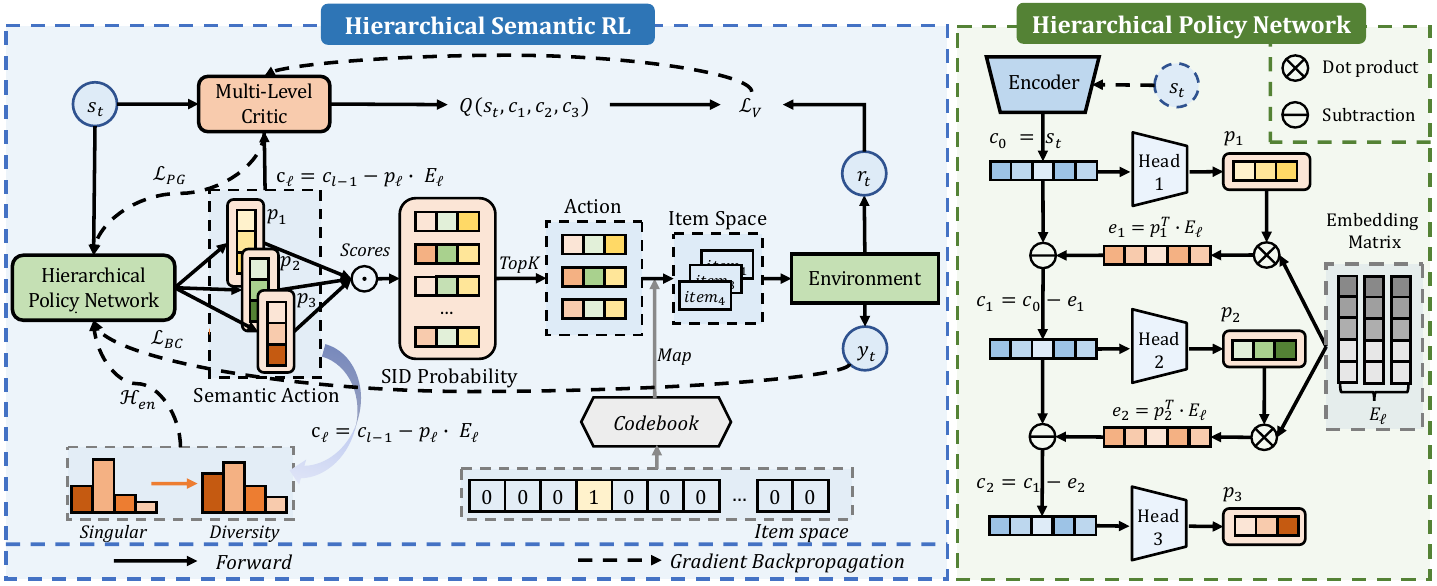}
    \caption{Overview of \textbf{HSRL}. A low-dimensional \textbf{Semantic Action Space (SAS)} maps items to fixed-length semantic IDs; a coarse-to-fine \textbf{Hierarchical Policy Network (HPN)} generates SID tokens autoregressively with residual context refinement, progressively narrowing the semantic subspace; a \textbf{Multi-Level Critic (MLC)} provides level-aware value estimation for structure-aware credit assignment; and a \textbf{joint actor–critic optimization} supports efficient \textbf{SID-based serving}.}
    \label{fig:maingraph}
\end{figure*}

This section presents the HSRL methodology. As illustrated in Figure~\ref{fig:maingraph}, our framework begins by constructing a fixed-dimensional Semantic Action Space (SAS) using SIDs, which decouples policy learning from the dynamic and ever-expanding item catalog and enables optimization in a fixed, low-dimensional space. To align with the intrinsic hierarchical structure of SIDs, the Hierarchical Policy Network (HPN) generates semantic actions autoregressively, refining decisions level-by-level through residual state modeling in a coarse-to-fine manner. To address the challenge of sparse, sequence-level rewards, the Multi-Level Critic (MLC) decomposes the global return into per-level value estimates, enabling structured credit assignment and enhancing training stability. Finally, a unified Actor–Critic optimization objective jointly trains the HPN and MLC, ensuring efficient and stable policy learning within the hierarchical semantic action space.

\subsection{Problem Formulation}

In this work, we cast the recommendation task as a reinforcement learning (RL) problem to optimize long-term user engagement rather than myopic single-click events. We model the interaction between the user and system as a Markov Decision Process (MDP):
\begin{equation}
\mathcal{M} = (\mathcal{S}, \mathcal{A}_{\mathrm{SAS}}, \mathcal{A}, \mathcal{P}, \mathcal{R}, \gamma),
\end{equation}

\begin{itemize}[leftmargin=*]

\item \textbf{State ($\mathcal{S}$).}
At step $t$, the state $s_t\in\mathcal{S}$ summarizes the user's context (e.g., profile and recent interactions).

\item \textbf{Semantic Action Space ($\mathcal{A}_{\mathrm{SAS}}$).}
The policy acts in a fixed-size hierarchical space induced by SIDs. A semantic action is a slate of $k$ length-$L$ token sequences:
\begin{equation}
\mathbf{Z}_t=[\mathbf{z}_{t,1},\ldots,\mathbf{z}_{t,k}], \qquad 
\mathbf{z}_{t,j}=[z_{t,j,1},\ldots,z_{t,j,L}],
\end{equation}
where $z_{t,j,\ell}\in\mathcal{V}_\ell$ for $\ell=1,\ldots,L$ and $j=1,\ldots,k$.

\item \textbf{Effect Action Space ($\mathcal{A}$).}
The environment executes an ordered list of $k$ items from the catalog $\mathcal{I}$:
\begin{equation}
\mathcal{A} = \mathcal{I}^{k}.
\end{equation}

\item \textbf{Policy and action realization ($\pi_\theta$).}
Given $s_t$, the policy samples a semantic action $\mathbf{Z}_t \sim \pi_\theta(\cdot \mid s_t)$, which is mapped to the effect action via a static lookup table:
\begin{equation}
\begin{aligned}
a_t &= \mathrm{Codebook}(\mathbf{Z}_t) \\
    &= [\,\mathrm{Codebook}(\mathbf{z}_{t,1}),\ldots,\mathrm{Codebook}(\mathbf{z}_{t,k})\,] \in \mathcal{A}.
\end{aligned}
\end{equation}

\item \textbf{Transition and Reward ($\mathcal{P},\mathcal{R}$).}
After executing $a_t$, the user environment returns feedback and the state evolves as:
\begin{equation}
\mathcal{P}(s_{t+1}\mid s_t,a_t), \qquad r_t=\mathcal{R}(s_t,a_t,s_{t+1}).
\end{equation}

\item \textbf{Objective ($\mathcal{J}$).}
For a trajectory $\tau=(s_0,a_0,\ldots,s_{T-1},a_{T-1})$,
\begin{equation}
\pi^\star \in \arg\max_{\pi_\theta}\ \mathbb{E}_{\tau \sim \pi_\theta,\mathcal{P}}\Big[\sum_{t=0}^{T-1} \gamma^t\, \mathcal{R}(s_t,a_t,s_{t+1})\Big].
\end{equation}
Here $\gamma\in(0,1)$ is the \emph{discount factor} that trades off immediate vs.\ future rewards. An episode terminates when the session ends or a maximum horizon $T_{\max}$ is reached.

\end{itemize}

\subsection{Semantic Action Space (SAS)}

\subsubsection{Reformulating Actions via Semantic IDs}

In industrial recommender systems, the item catalog is not only massive but also highly dynamic. New items are continuously introduced, existing ones disappear, and topical clusters evolve over time. Under the conventional item-as-action formulation, the action set $\mathcal{A}$ changes accordingly as the catalog updates. This non-stationarity forces the policy head to adapt to an ever-shifting target, causing previously learned logits to misalign with the current set of executable actions. As a result, both training stability and generalization suffer, particularly for newly added or long-tail items.

We address this by reformulating decisions in a Semantic Action Space (SAS) induced by SIDs. Each item is represented by a length-$L$ token sequence:
\begin{equation}
\mathbf{z}=[z_1,\ldots,z_L], \qquad z_\ell \in \mathcal{V}_\ell,\ \ell=1,\ldots,L,
\end{equation}
where $\{\mathcal{V}_\ell\}$ are \emph{fixed}, level-specific vocabularies. The policy acts in this stationary semantic space; the environment later realizes the effect action by deterministically decoding the SID to an item or an ordered slate. Crucially, as the catalog grows or churns, new items are indexed into the existing semantic coordinates rather than expanding or reshaping the policy output space. Thus the learned policy remains compatible with the evolving catalog without reparameterizing the action head.

Beyond stationarity, SAS provides structure that promotes robust generalization. The hierarchy induces a meaningful geometry in which semantically related items share long token prefixes. Gradients and value estimates can therefore propagate along shared semantic components, enabling the policy to transfer behavior from head to tail and to adapt to new items as soon as they are assigned SIDs. The same structure naturally supports \emph{coarse-to-fine} exploration: higher levels commit to broad intent, while lower levels refine towards specific attributes, yielding stable behavior under feedback. In short, SAS decouples policy learning from catalog dynamics while endowing the action space with interpretable, hierarchical semantics that improve stability and long-tail performance.

\subsubsection{Semantic ID Construction}
To instantiate the stationary Semantic Action Space introduced above, we build Semantic IDs offline and keep the resulting codebooks fixed during RL. Starting from item representations, we learn $L$ level-wise codebooks using residual-quantization $k$-means (RQ-$k$-means; see Appendix~\ref{app:rqkmeans}). For each level $\ell\in\{1,\ldots,L\}$, this yields a fixed vocabulary $\mathcal{V}_\ell$ with size $T_\ell=|\mathcal{V}_\ell|$. Each catalog item $i\in\mathcal{I}$ is then assigned one token per level by nearest-centroid lookup, producing its SID:
\begin{equation}
\mathbf{z}(i) = [\,z_{1}(i),\,z_{2}(i),\,\ldots,\,z_{L}(i)\,],\qquad z_{\ell}(i)\in\mathcal{V}_\ell.
\end{equation}
At interaction time $t$, the policy outputs a semantic action $\mathbf{z}_t=[z_{t,1},\ldots,z_{t,L}]$ in the \emph{same} fixed space, and the environment realizes the effect action via deterministic decoding:
\begin{equation}
a_t=\mathrm{Codebook}(\mathbf{z}_t)\in\mathcal{A}.
\end{equation}

As the catalog evolves, newly introduced items can be directly assigned SIDs using the fixed codebooks, while removed items only require deleting their SID-to-item entries. Since the level-wise vocabularies ${\mathcal{V}_\ell}$ and codebooks remain unchanged throughout RL, catalog churn does not alter the semantic decision space or the policy’s output dimensionality, enabling immediate compatibility with new items and stable execution for existing items.

\subsection{Hierarchical Policy Network (HPN)}

The Hierarchical Policy Network (HPN) is explicitly designed to align the policy’s decision process with the intrinsic generation structure of SIDs. This alignment is critical for correctly modeling the sequential and residual dependencies inherent in the Semantic Action Space, which is the primary architectural challenge in deploying SIDs within an RL framework. By strictly adhering to the SID's structure, the HPN operates under a coarse-to-fine reasoning mechanism, ensuring that the selection of higher-level tokens dictates the search scope for lower-level tokens, thus preventing decisions for different levels from being treated equally.

\subsubsection{From State to Coarse-to-Fine Plan}

We explicitly model the policy as an autoregressive, multi-step sequential process that mirrors the construction of an item's SID, $\mathbf{z}=[z_1, \dots, z_L]$. Given the current user state $s$, we first encode the interaction history $s_{\mathrm{hist}}$ into the initial global context vector $c_0 \in \mathbb{R}^d$:
\begin{equation}
c_0=\mathrm{TransformerEncoder}(s,s_{\mathrm{hist}}).
\end{equation}
The HPN then generates tokens sequentially over $L$ semantic levels. At each level $\ell=1,\ldots,L$, a dedicated policy head predicts the token distribution over the vocabulary $\mathcal{V}_\ell$ conditioned on the residual context from the previous level:
\begin{equation}
p\!\left(z_{\ell}\mid c_{\ell-1}\right)
=\mathrm{softmax}\!\left(W_\ell\,c_{\ell-1}\right),
\end{equation}
where $W_\ell$ is the level-$\ell$ projection matrix. This staged factorization is not merely for complexity reduction; it enforces the necessary hierarchical constraint, forcing the policy to refine its commitment at each step based on the evolving semantic context.

\subsubsection{Hierarchical Residual State Modeling (HRSM)}

To ensure this sequential process is truly a progressive refinement rather than a series of independent decisions, we introduce Hierarchical Residual State Modeling (HRSM). After predicting the token distribution at level $\ell$, the context $c_{\ell-1}$ must be updated to explicitly filter out the semantics already committed, forcing subsequent layers to focus only on the residual state information remaining for a fine-grained decision. This mechanism directly addresses and models the "coarse-to-fine" dependency. We calculate the expected semantic embedding $e_\ell$ using a learnable level-wise embedding matrix $E_\ell = [\mathbf{e}_{\ell,1}, \dots, \mathbf{e}_{\ell,T_\ell}]^\top \in \mathbb{R}^{T_\ell \times d}$, where the $z$-th row $\mathbf{e}_{\ell,z}$ corresponds to the semantic token $z \in \mathcal{V}_\ell$ defined by the SID generation process:

\begin{equation}
e_\ell \;=\; \big(p(z_{\ell}\mid c_{\ell-1})\big)^\top E_\ell .
\end{equation}
The next-level context $c_{\ell}$ is then updated by subtracting the expected commitment and normalizing, a process that achieves a crucial context refinement:
\begin{equation}
c_\ell \;=\; \mathrm{LayerNorm}\!\big(c_{\ell-1}-e_\ell\big).
\end{equation}
This recursive context update ensures that the trajectory $\mathcal{C}(s)=\{c_0,c_1,\ldots,c_L\}$ records a coherent and progressively refined reasoning path. By employing this residual strategy, we effectively stabilize training and significantly reduce the representation-decision mismatch that plagues conventional methods, as the policy’s input at any stage is precisely the representation of the remaining decision space.

\subsubsection{Item Likelihood under the Structured Policy}

For a candidate item with SID $\mathbf{z}=[z_{1},\ldots,z_{L}]$, its policy likelihood factorizes along the levels of the decision trajectory $\mathcal{C}(s)$:
\begin{equation}
\pi_\theta(\mathbf{z}\mid s)\;=\;\prod_{\ell=1}^{L} p\!\big(z_{\ell}\mid c_{\ell-1}\big).
\end{equation}
 Therefore, the item's probability is an explicit function of its semantic composition, aligning policy output directly with the structured action space. During execution, the sampled semantic action $\mathbf{z}=[z_{1},\ldots,z_{L}]$ is deterministically decoded to the final effect action $a=\mathrm{Codebook}(\mathbf{z})$, as established by our MDP.

\subsection{Multi-Level Critic (MLC)}

The semantic action space, while powerful, introduces a key challenge: the single, sequence-level reward signal significantly exacerbates the problems of sparse rewards and the credit assignment problem. A positive or negative reward provides no information about which specific part of the semantic ID sequence contributed to the final outcome. A traditional Critic network that outputs a single scalar value $V(s)$ cannot disambiguate this, leading to unstable training in policy gradient estimation. To address this, we design a novel Multi-Level Critic (MLC) that provides a more granular valuation of the hierarchical decision-making process. The MLC performs value decomposition across the entire action generation sequence, enabling precise, hierarchical credit assignment.

\subsubsection{Per-level Values for Value Decomposition.}
Our Critic network outputs a distinct value for the state at each step of the hierarchical action generation. Using the context trajectory $\mathcal{C}(s) = \{c_0, c_1, \dots, c_L\}$ recorded by the Actor, the Critic assigns an expected future reward to each semantic context $c_l$:
\begin{equation}
V_\phi(s,l) = f_\phi(c_l), \quad l=0,\dots,L,
\end{equation}
yielding the vector of per-level values:
\begin{equation}
\mathbf{V}_\phi(s) = \big(V_\phi(s,0), V_\phi(s,1), \dots, V_\phi(s,L)\big).
\end{equation}
Each element $V_\phi(s,l)$ represents the expected cumulative reward assuming optimal future behavior starting from the decision point after $l$ tokens have been determined.

\subsubsection{Adaptive Value Aggregation via Learnable Weights.}
The hierarchical generation of semantic actions follows a coarse-to-fine paradigm: higher-level tokens capture broad user intent, while lower-level tokens refine toward specific attributes. Crucially, under our Hierarchical Residual State Modeling, the context vector $c_\ell$ at level $\ell$ is obtained by subtracting the committed semantics from the previous context. As a result, deeper levels operate on increasingly specialized residual states, whereas earlier levels may retain richer global information about long-term user utility. This asymmetry implies that value estimates from different levels vary in reliability and scope. To leverage this structure, we introduce learnable importance weights $\{w_l\}$ that adaptively fuse per-level value estimates into a single scalar for TD learning:
\begin{equation}
\tilde{w}_l = \frac{\exp(w_l)}{\sum_{j=0}^{L} \exp(w_j)}, \qquad 
\hat V_\phi(s) = \sum_{l=0}^{L} \tilde{w}_l V_\phi(s,l).
\end{equation}
The weights are optimized end-to-end with the critic, allowing the model to automatically prioritize levels that are most predictive of the final sequence-level reward. This adaptive aggregation enables more accurate value estimation and robust hierarchical credit assignment across long decision chains.

\subsection{Hierarchical SID Actor-Critic Optimization}

To jointly optimize our Hierarchical Policy Network (HPN) and Multi-Level Critic (MLC), we adopt a unified Actor-Critic framework, specifically tailored for our structured action space and multi-level value functions. This unified approach ensures both networks learn synergistically, with the Critic's refined value estimates guiding the Actor toward more effective policies and stable convergence.

\subsubsection{TD Target and Advantage Function}

The foundation of our learning signal is the temporal-difference (TD) target, which estimates the action-value function $\hat{Q}(s_t,\mathbf{a}_t)$ by incorporating the immediate reward with the discounted value of the next state, where the dependence on $\mathbf{a}_t$ is captured by the observed $r_t$ and transition to $s_{t+1}$:
\begin{equation}
\hat{Q}(s_t, \mathbf{a}_t) = r_t + \gamma  \hat{V}_{\phi'}(s_{t+1}),
\end{equation}
where $r_t$ is the observed reward, $\gamma$ is the discount factor and $\hat V_{\phi'}(s_{t+1})$ is the weighted state value of the next state predicted by the target Multi-Level Critic.

The advantage function ${A}_t$ measures the improvement (or deficiency) of the chosen action relative to the expected outcome of the current state. It is the core signal used to update the policy:
\begin{equation}
{A}_t = \hat{Q}(s_t, \mathbf{a}_t) - \hat{V}_\phi(s_t).
\end{equation}
For practical gradient stability, we clip the advantage signal to $\mathrm{clip}({A}_t, -1, 1)$ before using it in the final policy update, preventing large advantage values from destabilizing training.

\subsubsection{Critic Objectives}


The Critic is trained to minimize the mean squared error (MSE) between its predicted weighted value $\hat V_\phi(s_t)$ and the calculated TD target $\hat{Q}(s_t,\mathbf{a}_t)$. We define the overall loss as the expected squared error over the data distribution $\mathcal{D}$:
\begin{equation}
\mathcal{L}_V(\phi) = \mathbb{E}_{s_t, \mathbf{a}_t \sim \mathcal{D}} \left[ \big(\hat{V}_\phi(s_t) - \hat{Q}(s_t, \mathbf{a}_t)\big)^2 \right]
\end{equation}
This ensures that the MLC produces accurate value estimates, serving as a stable reference for policy optimization.

\subsubsection{Actor Objectives}

The Actor learns a hierarchical SID generation policy by maximizing the expected advantage, while maintaining sufficient exploration and leveraging logged user feedback for stable optimization. For consistency with the Critic training, we compute all actor-side objectives on logged transitions $(s_t,\mathbf a_t,r_t,s_{t+1},\mathbf y_t)\sim\mathcal D$, where $\mathbf a_t=\{a_{t,1},\ldots,a_{t,k}\}$ is the displayed slate and $\mathbf y_t=\{y_{t,1},\ldots,y_{t,k}\}$ is the corresponding binary feedback.

\textbf{(1) Policy Gradient Loss.}
We adopt the advantage actor-critic objective and apply a slate-level advantage to the entire SID token sequence:
\begin{equation}
\ell_{\text{PG}}(s_t,\mathbf{a}_t)
= -\,\tilde{A}_t \,\log \pi_\theta\!\big(\mathbf{Z}(\mathbf{a}_t) \mid s_t\big),
\end{equation}
where $\tilde{A}_t=\mathrm{clip}(A_t,-1,1)$ stabilizes gradients, and
$\mathbf{Z}(\mathbf a_t)=\{\mathbf z_{a_{t,j}}\}_{j=1}^{k}$ denotes the SID token sequences of the slate items.

The log-likelihood factorizes across items and levels as:
\begin{equation}
\log \pi_\theta\!\big(\mathbf{Z}(\mathbf a_t)\mid s_t\big)
= \frac{1}{k}\sum_{j=1}^{k}\sum_{\ell=1}^{L}
\log p_\theta\!\big(z_{\ell}(a_{t,j}) \mid c_{\ell-1}^{(t,j)},\, s_t\big),
\end{equation}
where $c_{\ell-1}^{(t,j)}$ denotes the residual state at level $\ell\!-\!1$ for item $a_{t,j}$.

\textbf{(2) Entropy Regularization.}
To avoid premature convergence and encourage diverse semantic token exploration, we regularize the conditional token distributions by the average entropy:
\begin{equation}
\scalebox{0.85}{$
\begin{aligned}
\mathcal{H}_{\text{en}}(s_t,\mathbf a_t)
&= \frac{1}{k}\sum_{j=1}^{k}\sum_{\ell=1}^{L}
\Bigg(
-\sum_{z \in \mathcal{V}_\ell}
p_\theta\!\big(z \mid c_{\ell-1}^{(t,j)}, s_t\big)\,
\log p_\theta\!\big(z \mid c_{\ell-1}^{(t,j)}, s_t\big)
\Bigg).
\end{aligned}
$}
\end{equation}

\textbf{(3) Behavioral Cloning Loss.}
To ground policy learning in actual user preferences, we introduce a behavioral cloning term that encourages the policy to imitate items with positive feedback:
\begin{equation}
\ell_{\text{BC}}(s_t,\mathbf a_t) = -\frac{1}{\sum_{j=1}^{k} y_{t,j} + \epsilon} \sum_{j=1}^{k} y_{t,j} \log \pi_\theta\!\big(\mathbf z_{a_{t,j}} \mid s_t\big),
\end{equation}
where $y_{t,j} \in \{0,1\}$ denotes the binary feedback for item $a_{t,j}$ within the logged slate $\mathbf a_t$.

\textbf{Overall Actor Objective.} We denote $\mathbb{E}_{\mathcal D}[\cdot]$ as the expectation over
$(s_t,\mathbf a_t,r_t,s_{t+1},\mathbf y_t)\sim\mathcal D$.
\begin{equation}
\mathcal{L}_\pi(\theta)
= \mathbb{E}_{\mathcal D}\Big[
\ell_{\text{PG}}(s_t,\mathbf a_t)
- \lambda_{\text{en}}\,\mathcal{H}_{\text{en}}(s_t,\mathbf a_t)
+ \lambda_{\text{BC}}\,\ell_{\text{BC}}(s_t,\mathbf a_t)
\Big].
\end{equation}

\subsubsection{Joint Optimization}
Finally, we jointly optimize the Actor and Critic parameters by minimizing the combined objective:
\begin{equation}
\min_{\theta,\phi} \quad \mathcal{L}_V(\phi) + \mathcal{L}_\pi(\theta).
\end{equation}

\section{Experiments}
\subsection{Experimental Setup}

\subsubsection{Dataset.}

\begin{table}[t]
\centering
\begin{threeparttable}
\caption{Statistics of the datasets.}
\label{tab:dataset}
\footnotesize
\setlength{\tabcolsep}{4pt}
\begin{tabular}{lcccc}
\toprule
\textbf{Dataset} & \textbf{\#Users ($|U|$)} & \textbf{\#Items ($|I|$)} & \textbf{\#Records} & \textbf{\#Slate Size ($k$)} \\
\midrule
RL4RS          & --     & 283    & 0.78M  & 9  \\
MovieLens-1M   & 6.4K   & 3.7K   & 1.00M  & 10 \\
Kuaishou-Ads   & 90M    & 100K   & 450M   & 1  \\
\bottomrule
\end{tabular}
\begin{tablenotes}
\footnotesize

\item For the industrial dataset (\textit{Kuaishou-Ads}), we report \textbf{average daily statistics}, as both users and items exhibit daily dynamics.
\end{tablenotes}
\end{threeparttable}
\end{table}

We adopt two widely used public benchmarks, including \emph{RL4RS}~\cite{rl4rs} and \emph{MovieLens-1M}~\cite{ml1m}, following the same preprocessing pipeline as prior work~\cite{hac}. Both datasets are converted into a unified sequential format, with each timestep recording user interactions in chronological order. The earliest 80\% of logs are used for training, and the latest 20\% for evaluation. In addition to the public corpora, we further evaluate our framework on an industrial-scale dataset constructed from impression logs of Kuaishou’s internal advertising platform. This industrial dataset complements the public benchmarks, providing a rigorous validation of the scalability and robustness of our method under real-world commercial recommendation scenarios. Dataset statistics, including $|U|$, $|I|$, the number of records, and slate size $k$, are summarized in Table~\ref{tab:dataset}.

\subsubsection{Evaluation Metrics.}

On offline datasets, we evaluate per session long horizon performance in the simulator using \textbf{Total Reward} and \textbf{Depth}. Total Reward sums stepwise rewards derived from user click signals across the session, and Depth counts the realized number of interaction steps. An interaction ends when a leave signal is triggered, implemented as a patience variable that decays with consecutive unrewarded steps and terminates the episode once it falls below a threshold; Total Reward and Depth are then recorded. For online A/B testing, we report \textbf{Advertiser Value (ADVV)} and \textbf{Revenue}. ADVV measures the advertiser value, i.e., the amount the platform can reasonably charge the advertiser for the delivered traffic, while Revenue is the platform's realized income. The desired outcome is to increase both metrics, while ensuring ADVV grows faster than Revenue.

\subsubsection{Compared Methods.}
To provide a comprehensive evaluation, we compare our HSRL framework against a diverse set of state-of-the-art baselines from both supervised and reinforcement learning paradigms. All methods share identical feature encoders, candidate pools, data splits, and evaluation protocols to ensure a fair comparison. The compared methods are as follows:

\begin{itemize}[leftmargin=*]
    
\item \textbf{SL}: A supervised learning baseline using the SASRec~\cite{Kang2018SASRec} architecture, optimized to predict the next user interaction based on historical sequences. 

\item \textbf{A2C}: A synchronized variant of the Advantage Actor-Critic~\cite{a3c} that performs policy gradients directly in the discrete item-level (effect-action) space. 

\item \textbf{DDPG}: A deterministic policy gradient algorithm~\cite{ddpg} that operates in a continuous hyper-action space to handle large-scale action selection. 

\item \textbf{TD3}: An advanced actor-critic method~\cite{td3} that employs twin critics and delayed policy updates to mitigate value overestimation and stabilize training. 

\item \textbf{DDPG-RA}: A regularized DDPG variant~\cite{ddpgra} that learns continuous embeddings for discrete actions, allowing the actor to operate in a latent space while mapping back to the item space. 

\item \textbf{KGCR}: A hierarchical RL framework~\cite{complement1} that utilizes decision decomposition and integrates explicit user feedback signals for real-time policy refinement. 

\item \textbf{HG}: A hierarchical architecture~\cite{complement2} that combines reinforcement learning with multi-level expert feedback to enhance exploration efficiency in sparse-reward environments. 

\item \textbf{CHIRP}: A continual RL approach~\cite{complement3} that leverages state abstractions and symbolic representations to learn reusable and composable options. 

\item \textbf{HAC}: A hierarchical method~\cite{hac} that bridges hyper-action inference and effect-action selection through alignment and supervision, designed for large and dynamic candidate pools. 
    
\end{itemize}

\subsection{Overall Performance} Table~\ref{tab:offline_performance} summarizes the offline performance on RL4RS and MovieLens-1M. HSRL consistently achieves state-of-the-art results across both datasets in terms of \emph{Total Reward} and \emph{Depth}, which respectively represent long-term user engagement and interaction persistence. Specifically, on RL4RS, HSRL yields a \emph{Total Reward} of 12.013 and a \emph{Depth} of 12.841, outperforming the strongest baseline HAC by 13.4\% and 10.9\%. On MovieLens-1M, HSRL attains 18.773 in \emph{Total Reward} and 18.839 in \emph{Depth}, surpassing the strongest baseline CHIRP by 6.7\% and 5.0\%. These results indicate that conventional RL approaches operating over flat item-level spaces are fundamentally challenged by large action sets, which results in inefficient exploration. Furthermore, many continuous-action RL baselines relax discrete recommendations into continuous embeddings and project them back to items during execution. This process introduces a significant train-inference discrepancy that leads to sub-optimal long-term behaviors. In contrast, HSRL operates in a Semantic Action Space to reduce decision complexity and narrow the gap between learned actions and executed recommendations, thereby enabling more effective long-term optimization.

\begin{table}[!t]
\centering
\caption{Offline performance comparison.}
\label{tab:offline_performance}
\footnotesize
\setlength{\tabcolsep}{6pt}
\begin{tabular}{lcc|cc}
\toprule
\multirow{2}{*}{Model} & \multicolumn{2}{c|}{RL4RS} & \multicolumn{2}{c}{ML1M} \\
& Total Reward & Depth & Total Reward & Depth \\
\midrule
SL & 6.721 & 8.163 & 16.666 & 16.953 \\
\midrule
A2C    & 7.789 & 9.140 & 9.533  & 10.644 \\
DDPG   & 8.337 & 9.588 & 15.773 & 16.154 \\
TD3    & 8.553 & 9.791 & 17.465 & 17.677 \\
DDPG-RA  & 8.561 & 9.728 & 14.032 & 14.601 \\
HG     & 8.573 & 9.865 & 14.682 & 15.417 \\
KGCR   & 9.470 & 10.758 & 17.423 & 17.517 \\
CHIRP  & 9.841 & 10.934 & \underline{17.591} & \underline{17.951} \\
HAC & \underline{10.595} & \underline{11.582} & 17.587 & 17.787 \\

\midrule
HSRL (Ours) & \textbf{12.013} & \textbf{12.841} & \textbf{18.773} & \textbf{18.839} \\

\bottomrule
\end{tabular}
\end{table}

\subsection{Ablation Study}

\begin{table}[t]
\centering
\caption{Ablation study results.}
\label{tab:ablation}
\footnotesize
\setlength{\tabcolsep}{5pt}
\begin{tabular}{lcc|cc}
\toprule
\multirow{2}{*}{Model Variant} & \multicolumn{2}{c|}{Scores} & \multicolumn{2}{c}{Drop (\%)} \\
& Total Reward & Depth & Reward $\Delta$ & Depth $\Delta$ \\
\midrule
HSRL (Full) & \textbf{12.013} & \textbf{12.841} & -- & -- \\
\midrule
w/o Entropy             & 9.766  & 10.843  & $-18.7$ & $-15.6$ \\
w/o Hierarchical Policy & 10.175 & 11.204  & $-15.3$ & $-12.7$ \\
w/o BC loss             & 11.254 & 12.123  & $-6.3$  & $-5.6$  \\
w/o Multi-Level Critic  & 10.734 & 11.756  & $-10.6$ & $-8.4$  \\
\bottomrule
\end{tabular}
\end{table}

We evaluate the contribution of HSRL's key components on the RL4RS dataset, as summarized in Table~\ref{tab:ablation}. Removing entropy regularization or the hierarchical policy structure leads to the most substantial performance drops, where \emph{Total Reward} decreases by 18.7\% and 15.3\% respectively. This underscores the importance of maintaining exploration pressure and leveraging coarse-to-fine decision making within the Semantic Action Space. Additionally, the degradation caused by disabling the multi-level critic suggests that token-level value estimation is essential for accurate credit assignment. Lastly, the moderate decline observed without the behavioral cloning term indicates its role in stabilizing training dynamics. Overall, these results confirm that the structural design of HSRL is fundamental to its performance gains.

\subsection{Online Deployment}

\begin{figure}
    \centering
    \includegraphics[width=0.8\linewidth]{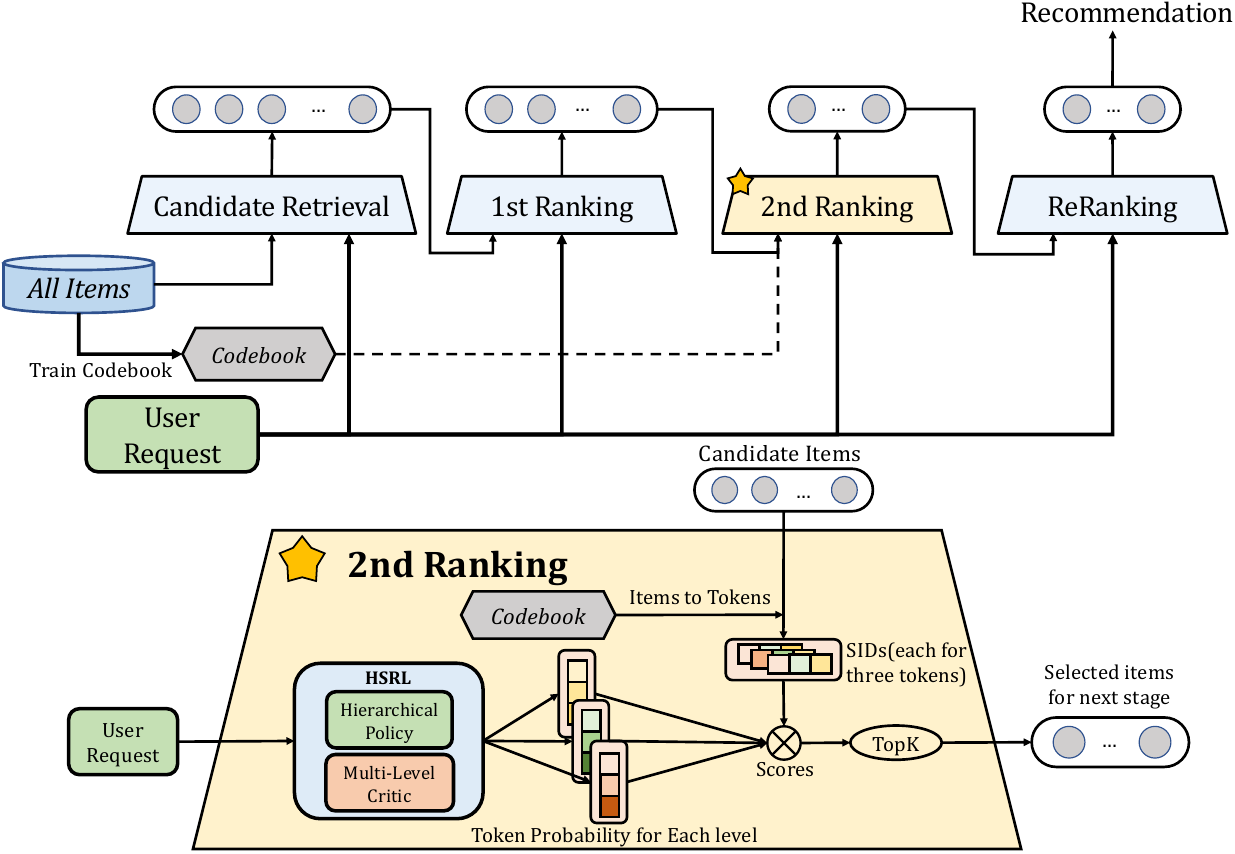}
    \caption{Online Deployment.}
    \label{fig:online}
\end{figure}

\subsubsection{Deployment Pipeline.}

As shown in Figure~\ref{fig:online}, HSRL is integrated into Kuaishou's advertising system and deployed in the second-stage ranking module. Offline, we construct a fixed SID codebook that maps each item in the full catalog to a deterministic 3-token SID. At serving time, the upstream recall module provides a limited candidate set for each request. Given the user request, HSRL performs a single forward pass to generate three level-wise token distributions. Each candidate item is then decoded into its 3-token SID using the precomputed codebook, and its final score is obtained by multiplying the probabilities of its corresponding tokens under the three predicted distributions. Candidates are ranked by this composite score, and the top-ranked items are forwarded to the subsequent stage. This design preserves the semantic structure of the Semantic Action Space and satisfies strict latency requirements in production. Notably, HSRL reuses the exact same backbone and feature interface as the supervised second-stage ranker, so the serving graph is unchanged and the shift is purely from SL to RL, with no added online latency.
\subsubsection{Online Performance.}

\begin{figure}
    \centering
    \includegraphics[width=0.7\linewidth]{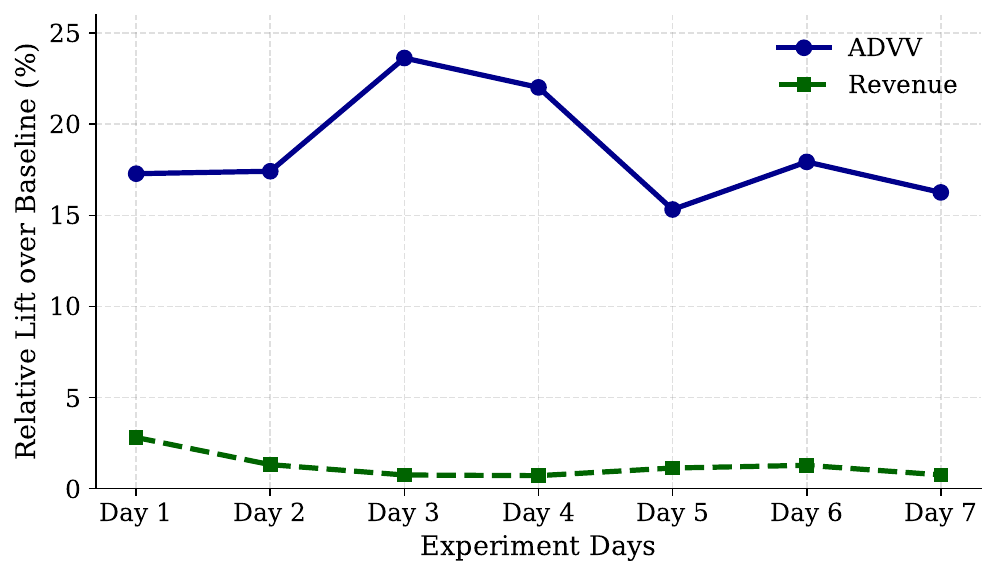}
    \caption{Online Performance.}
    \label{ab_test_results}
\end{figure}

As shown in Figure~\ref{ab_test_results}, the online A/B testing in a live production environment provides stronger evidence than simulator results and further validates the effectiveness of our model. We conducted a 7-day online A/B testing to evaluate the real-world business impact of our HSRL framework. Both the control and experimental groups were allocated 15\% of total traffic, ensuring a fair and large-scale comparison. For the control group, we deployed a supervised learning (SL) baseline. HSRL achieved significant gains on business metrics, improving ADVV by \textbf{18.421\%} and a \textbf{1.251\%} increase in Revenue. This suggests that advertiser value improves with a relatively smaller increase in advertiser spend, validating the model’s practical effectiveness. These results show that HSRL offers clear advantages for ad recommendation by optimizing long-term user value rather than short-horizon signals and demonstrate the practical feasibility of reinforcement learning in recommendation.

\subsection{Analysis of Multi-Level Critic Weights}

To validate the design of our Multi-Level Critic (MLC), we analyze the evolution of the learned importance weights $\{w_l\}$ during training. The weights stabilize over time, converging to a consistent distribution that reflects the model’s credit assignment strategy. Specifically, we observe the ordering $w_0 > w_1 > w_3 > w_2$, where $w_0$ corresponds to the initial context $c_0$, and $w_1, w_2, w_3$ correspond to the contexts after generating the first, second, and third tokens of the 3-token SID, respectively. The dominance of $w_0$ is expected, as it encodes the global user state and defines the initial feasible semantic subspace. The first token ($w_1$) refines this subspace by committing to a broad category, yielding the largest reduction in action space uncertainty and thus receiving high value attribution. In contrast, the second token ($w_2$) exhibits the lowest weight, suggesting limited incremental information gain. This aligns with the \emph{hourglass effect}~\cite{hourglass} in residual-quantized SIDs, where middle layers suffer from codeword concentration and low entropy, reducing their discriminative capacity. The final token ($w_3$), however, recovers discriminative power by specifying fine-grained attributes, resulting in a higher weight than $w_2$. Together, these findings confirm that the MLC adaptively allocates credit in a manner consistent with the hierarchical information structure of SIDs, enabling accurate and interpretable value decomposition. The full training dynamics of the weights are visualized in Appendix~\ref{app:mlc_weights}.

\subsection{Sensitivity analysis}

We conduct a sensitivity analysis on \textsc{ML1M} to assess robustness with respect to three hyperparameters of the semantic actionization: the entropy coefficient $\lambda_{\text{en}}$, the per-level vocabulary size $|\mathcal{V}_l|$, and the number of semantic levels $L$. In each study, we vary a single factor while holding the others fixed.

\subsubsection{Entropy Coefficient($\lambda_{\text{en}}$).}
First, we fix the vocabulary size at $|\mathcal{V}_l|{=}64$ and the number of semantic levels at $L{=}3$, and sweep the entropy coefficient. As shown in Figure~\ref{fig:entropy_size_sensitivity}, performance peaks at $\lambda_{\text{en}}{=}0.1$; setting $\lambda_{\text{en}}{=}0.2$ yields slightly lower scores, no entropy ($0.0$) is worse, and $\lambda_{\text{en}}{=}0.3$ is the lowest. Overall, a small amount of entropy encourages exploration in the low-dimensional, fixed SAS while preserving the HPN’s coarse-to-fine refinement, whereas too much entropy injects noise and blurs level-conditioned signals. In practice, $\lambda_{\text{en}}\approx 0.1$ offers a robust performance.

\subsubsection{Vocabulary Size ($|\mathcal{V}_l|$).}
Second, we fix the entropy coefficient at $\lambda_{\text{en}}{=}0.1$ and the number of levels at $L{=}3$, and vary $|\mathcal{V}_l|$. As shown in Figure~\ref{fig:sen_vocab}, the results exhibit a rise--then--fall pattern: performance improves steadily with larger codebooks and peaks around $|\mathcal{V}_l|\approx 80$, after which further enlarging the vocabulary yields diminishing returns and eventual degradation. Very small vocabularies under-represent semantics (tokens become overly broad), whereas very large vocabularies inflate the per-level action space and raise optimization difficulty without proportional gains. Overall, a sufficiently large but not overly complex codebook (around $80$) balances expressivity and learning efficiency.

\subsubsection{Number of Semantic Levels ($L$).}

Finally, we fix the vocabulary size at $|\mathcal{V}_l|{=}64$ and the entropy coefficient at $\lambda_{\text{en}}{=}0.1$, and vary $L$. As shown in Figure~\ref{fig:sen_size}, a moderate depth ($L{=}4$) achieves the best overall performance, with $L{=}5$ close behind. A shallower hierarchy ($L{=}3$) makes the semantic action space too coarse to capture fine-grained preferences, whereas an overly deep hierarchy increases exploration difficulty, complicates credit assignment, and aggravates lower-level data sparsity, which together reduce training stability and final quality. In practice, setting $L{\approx}4$ provides a strong trade-off between semantic precision and learnability.

\begin{figure}[t]
  \centering
  \begin{subfigure}[t]{0.33\linewidth}
    \centering
    \includegraphics[width=\linewidth]{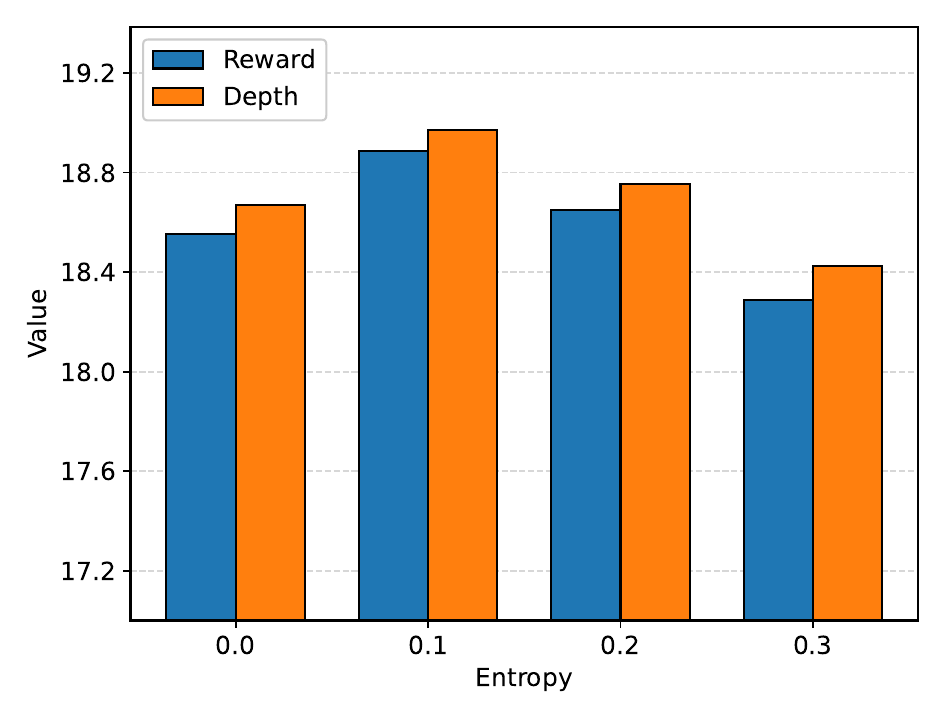}
    \caption{Entropy}
    \label{fig:entropy_size_sensitivity}
  \end{subfigure}\hfill
  \begin{subfigure}[t]{0.33\linewidth}
    \centering
    \includegraphics[width=\linewidth]{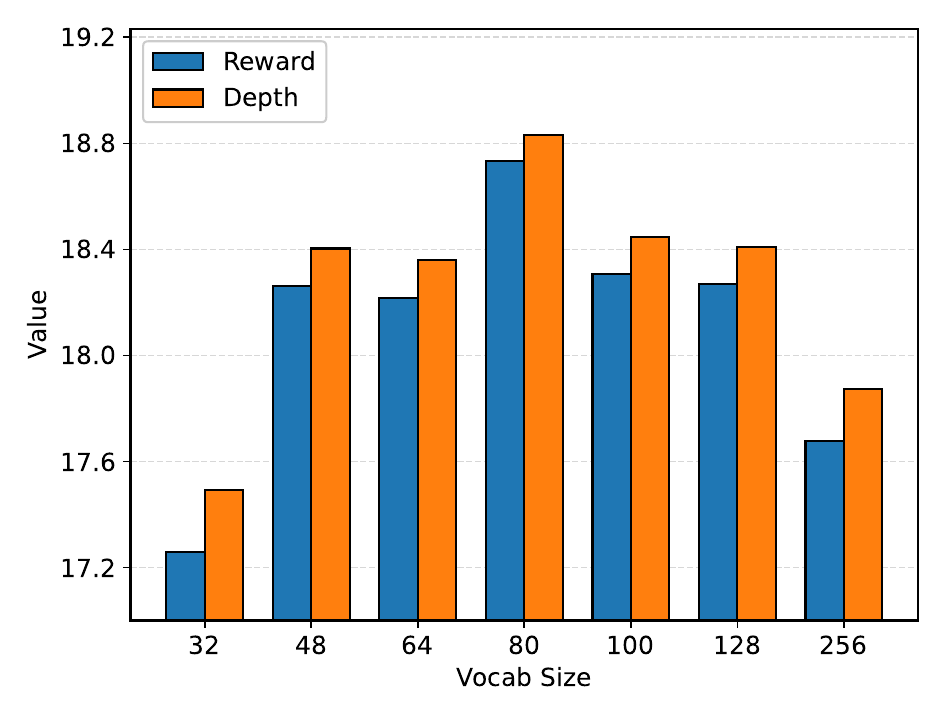}
    \caption{Vocab Size}
    \label{fig:sen_vocab}
  \end{subfigure}\hfill
  \begin{subfigure}[t]{0.33\linewidth}
    \centering
    \includegraphics[width=\linewidth]{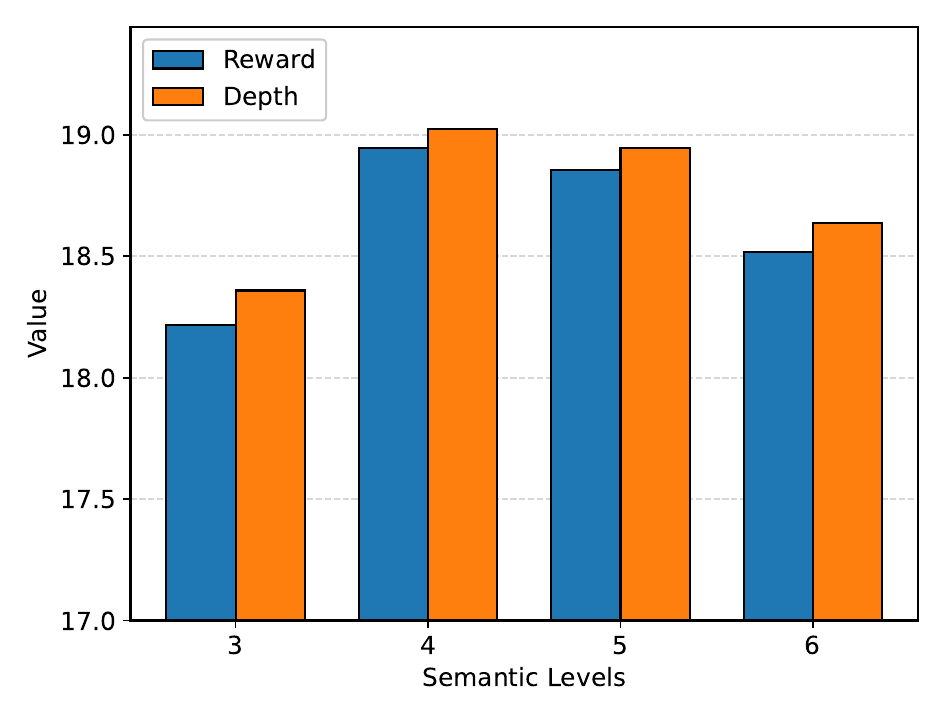}
    \caption{Semantic Levels}
    \label{fig:sen_size}
  \end{subfigure}
  \caption{Sensitivity Analysis.}
  \label{fig:sensitivity}
\end{figure}

\section{Conclusion}
\label{sec:conclusion}
In this paper, we address a core bottleneck in reinforcement learning based recommender systems: the dynamic, high-dimensional action space. We propose \textbf{Hierarchical Semantic Reinforcement Learning (HSRL)}, which introduces a \textbf{Semantic Action Space (SAS)} to decouple policy learning from the ever-changing item catalog by mapping items to fixed-dimensional Semantic IDs (SIDs). To exploit this structure, HSRL comprises two components: a \textbf{Hierarchical Policy Network (HPN)} that generates SIDs in a coarse-to-fine manner with residual state modeling to ensure structured decisions and generalization, and a \textbf{Multi-Level Critic (MLC)} that performs token-level value estimation to enable fine-grained credit assignment under sparse, sequence-level rewards. Experiments on public benchmarks and a large-scale industrial dataset show consistent gains on long-term engagement metrics. In a seven-day online A/B testing on a commercial advertising platform, HSRL achieved a \textbf{18.421\%} improvement in ADVV and a \textbf{1.251\%} increase in Revenue, demonstrating strong real-world effectiveness. Overall, HSRL provides a scalable and stable framework for RL-based recommendation and establishes structured semantic action modeling as a practical paradigm for large-scale systems.

\section{Acknowledgments}

This study was partially funded by the National Natural Science Foundation of China (62502340, 72471165), Beijing-Tianjin-Hebei Natural Science Foundation Cooperation Project (25JJJJC0045), and Tianjin University Independent Innovation Fund-Social Impact Project (2025XSC-0048).

\bibliographystyle{ACM-Reference-Format}
\bibliography{arxiv}

\appendix

\section{Dataset Processing}
\label{app:dataset}

For the ML-1M dataset, we treat movies with user ratings higher than 3 as positive samples (indicating a “like”) and all others as negative samples. Each user’s interaction sequence is segmented chronologically into subsequences of length 10. For each such subsequence, only the positive interactions that occurred prior to it are included in the historical behavior. The final dataset consists of records in the format: (user ID, historical behavior sequence, current item list, label list), which aligns with the RL4RS dataset schema.

\section{Implementation Details.}
The interactive environment follows the same construction: for each dataset we train a user response model $\Psi: \mathcal{S}\times\mathcal{A}\rightarrow \mathbb{R}^k$ that maps a state (built from static user features and dynamic histories) and a recommended slate to click-through probabilities, from which binary feedback $y_t\in\{0,1\}^k$ is sampled. Rewards are defined as the average item-wise signal, bounded in $[{-}0.2,\,1.0]$. Our actor backbone is \textit{SASRec}\cite{Kang2018SASRec}, as in the prior setup. We train one user-response simulator $\Psi$ on the training split (used during policy learning) and a second simulator pretrained on the entire dataset (used only for evaluation). Policies are optimized in the first simulator and assessed in the second. Unless noted, we fix the discount at $\gamma=0.9$ and cap the interaction horizon at $20$ steps; in practice, RL methods stabilize within $\sim 50{,}000$ iterations. Our code is released at \url{https://github.com/MinmaoWang/HSRL}.

\section{RQ-$k$-means for SID Tokenization}
\label{app:rqkmeans}

We use residual-quantization $k$-means (RQ-$k$-means) to build level-wise codebooks and assign SIDs in a coarse-to-fine manner. Let $x_i\in\mathbb{R}^d$ be the embedding of item $i\in\mathcal{I}$. We construct $L$ levels with vocabulary sizes $T_\ell=|\mathcal{V}_\ell|$.

\paragraph{Initialization.}
Set the initial residuals as the item embeddings:
\begin{equation}
R^{(1)} := X = \big[x_1^\top;\ldots; x_{|\mathcal{I}|}^\top\big]\in\mathbb{R}^{|\mathcal{I}|\times d},\quad
\mathbf{r}^{(1)}_i := x_i .
\end{equation}

\paragraph{Level $\ell=1,\ldots,L$: codebook learning and assignment.}
Cluster current residuals to obtain the level-$\ell$ codebook (centroids):
\begin{equation}
C^{(\ell)}=\mathrm{k\text{-}means}\!\big(R^{(\ell)},\,T_\ell\big), \qquad
C^{(\ell)}=\{\mathbf{c}^{(\ell)}_{k}\in\mathbb{R}^{d}\mid k=1,\ldots,T_\ell\}.
\end{equation}
Assign item $i$ to its nearest centroid (the token index \(z_\ell(i)\in\mathcal{V}_\ell=\{1,\ldots,T_\ell\}\)):
\begin{equation}
z_\ell(i)=\arg\min_{k\in\{1,\ldots,T_\ell\}} \big\|\,\mathbf{r}^{(\ell)}_i-\mathbf{c}^{(\ell)}_{k}\,\big\|_2 .
\end{equation}
Update the residual for the next level:
\begin{equation}
\mathbf{r}^{(\ell+1)}_i=\mathbf{r}^{(\ell)}_i-\mathbf{c}^{(\ell)}_{\,z_\ell(i)} .
\end{equation}

\paragraph{Output.}
After $L$ levels, the Semantic ID of item $i$ is
\begin{equation}
\mathbf{z}(i) = [\,z_1(i),\,z_2(i),\,\ldots,\,z_L(i)\,], \qquad z_\ell(i)\in\mathcal{V}_\ell .
\end{equation}

\section{Multi-Level Critic Weight Dynamics}
\label{app:mlc_weights}

Figure~\ref{fig:appendix_mlc_weights} shows the evolution of the learned importance weights $\{w_l\}$ during training on the RL4RS dataset. The weights converge to a stable ordering $w_0 > w_1 > w_3 > w_2$, consistent with the analysis in Section~4.5.

\begin{figure}[h]
    \centering
    \includegraphics[width=0.9\linewidth]{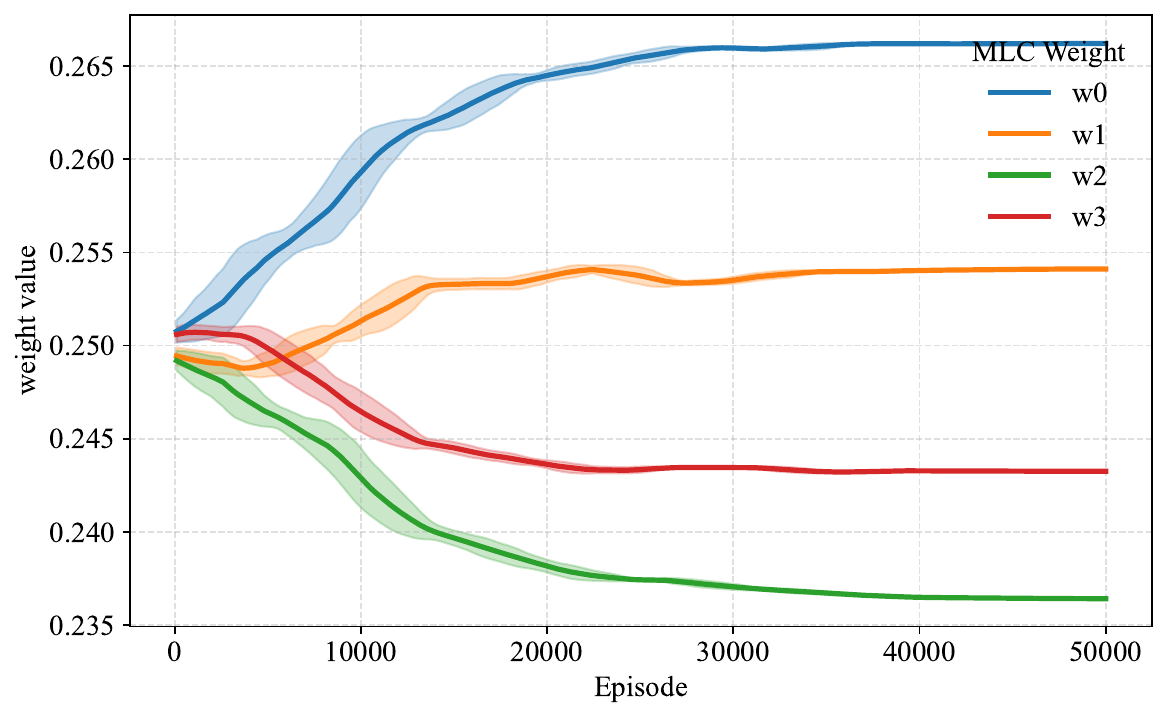}
    \caption{Training dynamics of MLC importance weights.}
    \label{fig:appendix_mlc_weights}
\end{figure}

\section{Deterministic SAS-to-Item Mapping}
Given a policy output $\mathbf{z}$ in the fixed Semantic Action Space, we \emph{directly} decode it via a static codebook $\mathrm{Codebook}(\mathbf{z})$ to the item set
$S(\mathbf{z})=\{\,i \mid \mathbf{z}(i)=\mathbf{z}\,\}$, \textbf{without any similarity computation or NN retrieval}. We minimize SID collisions offline; when rare collisions remain ($|S(\mathbf{z})|>1$), we treat them as equivalent candidates under the same semantic intent: include all if capacity allows, or apply lightweight business priors (e.g., recency/quality/freq-capping) to select or diversify when capacity is constrained. Thus the SAS$\!\to\!$Item mapping is a deterministic lookup rather than a similarity-based approximation.

\section{Theoretical Analysis: Structural Alignment of HPN with the SID Generation Process}
\label{app:theory_consistency}

In this section, we provide a principled justification for the hierarchical design of our policy network (HPN).  
Rather than asserting strict equivalence, we demonstrate that HPN is \textbf{structurally aligned} with the residual quantization process that underlies Semantic ID (SID) generation.  
This alignment provides a theoretical basis for the hierarchical, residual architecture adopted in our policy design.

\subsection{Preliminaries: Hierarchical Structure of SID Generation}

The Semantic ID (SID) of an item is generated through a residual quantization process over $L$ levels.  
Let $x_i \in \mathbb{R}^d$ denote the embedding of item $i$.  
At each level $\ell \in \{1, \dots, L\}$, a fixed codebook  
$\mathcal{C}^{(\ell)} = \{\mathbf{c}^{(\ell)}_1, \dots, \mathbf{c}^{(\ell)}_{T_\ell}\}$  
is learned offline via $k$-means.  
Given the residual vector $\mathbf{r}^{(\ell)}$, the token is selected as:
\begin{equation}
z_\ell^\star = \arg\min_{k \in \{1, \dots, T_\ell\}} 
\|\mathbf{r}^{(\ell)} - \mathbf{c}^{(\ell)}_k\|_2^2.
\label{eq:sid-selection}
\end{equation}

The residual for the next level is updated as: 
\begin{equation}
\mathbf{r}^{(\ell+1)} = \mathbf{r}^{(\ell)} - \mathbf{c}^{(\ell)}_{z_\ell^\star}.
\label{eq:sid-residual}
\end{equation}

This defines a hierarchical, autoregressive mapping with conditional dependencies:
$p(z_1, \dots, z_L) = \prod_{\ell=1}^{L} p(z_\ell \mid z_{<\ell})$.

\subsection{HPN as a Continuous, Differentiable Approximation}

The Hierarchical Policy Network (HPN) models the same autoregressive dependency structure in a continuous and differentiable manner.  
At level $\ell$, HPN predicts a token distribution:
\begin{equation}
p_\theta(z_\ell \mid \mathbf{c}_{\ell-1}) 
= \mathrm{softmax}(W_\ell \mathbf{c}_{\ell-1}),
\label{eq:hpn-policy}
\end{equation}

and computes an expected semantic embedding using a learnable embedding matrix $E_\ell = [\mathbf{e}_{\ell,1}, \dots, \mathbf{e}_{\ell,T_\ell}]^\top$, where the $z$-th row $\mathbf{e}_{\ell,z}$ corresponds to the semantic token $z \in \mathcal{V}_\ell$ defined by the SID generation process:
\begin{equation}
\mathbf{e}_\ell 
= \sum_{z \in \mathcal{V}_\ell} 
p_\theta(z_\ell = z \mid \mathbf{c}_{\ell-1}) \, \mathbf{e}_{\ell,z}.
\label{eq:hpn-expect}
\end{equation}

The context is updated as:
\begin{equation}
\mathbf{c}_\ell = 
\mathrm{LayerNorm}\big(\mathbf{c}_{\ell-1} - \mathbf{e}_\ell\big).
\label{eq:hpn-residual}
\end{equation}

Comparing Eq.~\eqref{eq:sid-residual} and Eq.~\eqref{eq:hpn-residual}, we observe that both processes follow the same residual refinement principle:  
the current state is updated by subtracting a representation of the chosen semantic component.  
In SID, this component is the discrete centroid $\mathbf{c}^{(\ell)}_{z_\ell^\star}$;  
in HPN, it is the continuous expectation $\mathbf{e}_\ell$, which can be viewed as a \textbf{soft relaxation} of the discrete choice.

\subsection{Structural Alignment Argument}

\begin{proposition}[Structural Alignment]
The HPN update rule (Eq.~\ref{eq:hpn-residual}) implements a continuous, differentiable approximation of the SID residual update (Eq.~\ref{eq:sid-residual}), preserving the coarse-to-fine autoregressive dependency structure. Specifically:
\begin{enumerate}[label=(\roman*)]
    \item Both processes are autoregressive: the decision at level $\ell$ depends on the state refined by levels $1$ to $\ell-1$;
    \item Both use a residual mechanism: the next-level state is derived by subtracting the semantic contribution of the current level;
    \item When the policy distribution $p_\theta(z_\ell \mid \mathbf{c}_{\ell-1})$ becomes concentrated around the optimal token $z_\ell^\star$, the expected embedding $\mathbf{e}_\ell$ aligns with the semantic role of $\mathbf{c}^{(\ell)}_{z_\ell^\star}$, and the HPN update approximates the SID update up to normalization.
\end{enumerate}
\end{proposition}

\begin{proof}
The first two points follow directly from Eq.~\eqref{eq:sid-residual} and Eq.~\eqref{eq:hpn-residual}.  
For (iii), when $p_\theta(z_\ell = z_\ell^\star \mid \mathbf{c}_{\ell-1}) \to 1$, the expected embedding becomes $\mathbf{e}_\ell \to \mathbf{e}_{\ell,z_\ell^\star}$. Although $\mathbf{e}_{\ell,z_\ell^\star}$ is not necessarily equal to $\mathbf{c}^{(\ell)}_{z_\ell^\star}$, both correspond to the same semantic token $z_\ell^\star$ in the shared token-ID space $\mathcal{V}_\ell$. Thus, the direction of the residual update $\mathbf{c}_{\ell-1} - \mathbf{e}_\ell$ remains semantically consistent with $\mathbf{r}^{(\ell)} - \mathbf{c}^{(\ell)}_{z_\ell^\star}$.  
Substituting into Eq.~\eqref{eq:hpn-residual} yields  
$\mathbf{c}_\ell \approx \mathrm{LayerNorm}(\mathbf{c}_{\ell-1} - \mathbf{e}_{\ell,z_\ell^\star})$.  
While LayerNorm introduces non-linear normalization, it preserves the directional semantics of the residual update and stabilizes magnitude during training.  
Furthermore, under typical training conditions, LayerNorm acts as an approximately non-expansive transformation, helping to bound gradient norms and reinforce stability.
\end{proof}

\subsection{Practical Implications}

This analysis shows that the HPN architecture is not heuristic but a \textbf{principled design choice} that mirrors the generative process of SIDs.  
The hierarchical residual structure ensures that:
\begin{enumerate}[label=(\roman*)]
    \item Policy decisions are made in a semantically coherent order (coarse $\rightarrow$ fine);
    \item The state representation at each level reflects the remaining decision space;
    \item Gradient signals propagate through a stable, approximately non-expanding residual path (empirically supported by Table~\ref{tab:ablation}).
\end{enumerate}
We emphasize that the alignment is structural rather than numerical: the goal is not to replicate the SID generation path exactly, but to impose an inductive bias that faithfully reflects the semantics of the action space through shared token indexing, thereby enhancing learning stability and generalization.

\section{Additional Variance Analysis}

\begin{table}[ht]
\centering
\caption{Detailed performance with variance analysis.}
\label{tab:variance_analysis}
\small
\begin{tabular}{l|cc|cc}
\hline
\multirow{2}{*}{Method} & \multicolumn{2}{c|}{RL4RS} & \multicolumn{2}{c}{ML1M} \\
& Reward & Depth & Reward & Depth \\
\hline

HG    & 8.57 $\pm$ 0.15 & 9.87 $\pm$ 0.10 & 14.68 $\pm$ 0.09 & 15.42 $\pm$ 0.15 \\
KGCR  & 9.47 $\pm$ 0.13 & 10.76 $\pm$ 0.12 & 17.42 $\pm$ 0.22 & 17.52 $\pm$ 0.28 \\
CHIRP & 9.84 $\pm$ 0.28 & 10.93 $\pm$ 0.40 & 17.59 $\pm$ 0.22 & 17.95 $\pm$ 0.06 \\

HAC    & 10.59 $\pm$ 0.18 & 11.58 $\pm$ 0.16 & 17.59 $\pm$ 0.26 & 17.79 $\pm$ 0.24 \\
Ours  & \textbf{12.01 $\pm$ 0.22} & \textbf{12.84 $\pm$ 0.29} & \textbf{18.77 $\pm$ 0.14} & \textbf{18.84 $\pm$ 0.17} \\
\hline
\end{tabular}
\end{table}

To validate the stability of our experimental findings, we report the mean performance and standard deviations for HSRL and representative competitive baselines in Table~\ref{tab:variance_analysis}. This analysis demonstrates that the performance improvements achieved by HSRL are statistically significant rather than being artifacts of stochastic variation. While all reinforcement learning methods exhibit inherent variance, the performance margin between HSRL and the strongest baselines consistently and substantially exceeds the corresponding standard deviations. For instance, on the RL4RS dataset, the lower bound of HSRL's performance (mean minus standard deviation) remains well above the upper bound of the best baseline. These results confirm that the superior performance stems from our architectural design rather than advantageous initializations, which underscores the consistent reliability of our approach across different trials.

\section{LLM Usage Statement}

In the preparation of this manuscript, a large language model (LLM), such as ChatGPT, was only used to assist with text editing and refinement. All writing, experimental design, analysis, and final interpretation of results were independently conducted by the author team, who bear full responsibility for the content and conclusions of the paper. All authors take full responsibility for the integrity and accuracy of the manuscript.

\end{document}